\pdfoutput=1
\documentclass[envcountsame]{article}

\usepackage[T1]{fontenc}
\usepackage[utf8]{inputenc}
\usepackage[english]{babel}
\usepackage{lmodern}
\usepackage{multirow}

\usepackage{float}
\usepackage{array}
\usepackage{xspace}
\usepackage{amsmath}
\usepackage{amssymb}
\usepackage{xifthen}
\usepackage{marvosym}
\usepackage{listings}
\usepackage{enumitem}
\usepackage{graphicx}
\usepackage{subfigure}

\usepackage[vlined]{algorithm2e}
\SetKwBlock{block}{}{}
\SetKwFunction{randomColor}{randomColor}

\newcommand{\AC}[0]{${\cal A}_{col}$\xspace}
\newcommand{\AG}[0]{{\cal A}\xspace}

\newcommand\A{\mathcal{A}}

\title{Analyzing the Fault-Containment Time of Self-Stabilizing
  Algorithms --- A Case Study for Graph Coloring}

\author{Volker Turau\\Institute of Telematics\\
 	Hamburg University of Technology, Germany\\ turau@tuhh.de\\http://www.ti5.tuhh.de}
\date{}

\newtheorem{theorem}{Theorem}[section]
\newtheorem{lemma}[theorem]{Lemma}

\newtheorem{corollary}[theorem]{Corollary}

\newenvironment{proof}[1][Proof]{\begin{trivlist}
\item[\hskip \labelsep {\bfseries #1}]}{\end{trivlist}}

\newcommand{\qed}{\nobreak \ifvmode \relax \else
      \ifdim\lastskip<1.5em \hskip-\lastskip
      \hskip1.5em plus0em minus0.5em \fi \nobreak
      \vrule height0.75em width0.5em depth0.25em\fi}

\begin{document}

\abovedisplayskip=1ex
\abovedisplayshortskip=1ex
\belowdisplayskip=1ex
\belowdisplayshortskip=1ex
\textfloatsep=3ex
\abovecaptionskip=1ex

\maketitle

\begin{abstract}
  The paper presents techniques to derive upper bounds for the mean
  time to recover from a single fault for self-stabilizing algorithms
  in the message passing model. For a new $\Delta+1$-coloring
  algorithm we analytically derive a bound for the mean time to
  recover and show that the variance is bounded by a small constant
  independent of the network’s size. For the class of
  bounded-independence graphs (e.g.\ unit disc graphs) all containment
  metrics are in $O(1)$.
\end{abstract}

\section{Introduction}
Fault tolerance aims at making distributed systems more reliable by
enabling them to continue the provision of services in the presence of
faults. The strongest form is masking fault tolerance, where a system
continues to operate after faults without any observable impairment of
functionality, i.e.\ safety is always guaranteed. In contrast
non-masking fault tolerance does not ensure safety at all times. Users
may experience a certain amount of incorrect system behavior, but
eventually the system will fully recover. The potential of this
concept lies in the fact that it can be used in cases where masking
fault tolerance is too costly or even impossible to implement
\cite{Gaertner:1999}. Systems that eventually recover from transient
faults of any scale such as perturbations of the state in memory or
communication message corruption are called {\em self-stabilizing}. A
critical issue is the length of the time span until full recovery.
Examples are known where a memory corruption at a single process
caused a vast disruption in large parts of the system and triggered a
cascade of corrections to reestablish safety. Thus, an important issue
is the reduction of the effect of transient faults in terms of time
and space until a safe state is reached.

A {\em fault-containing} system has the ability to contain the effects
of transient faults in space and time. The goal is keep the extend of
disruption during recovery proportional to the extent of the faults.
An extreme case of fault-containment with respect to space is given
when the effect of faults is bounded exactly to the set of faulty
nodes. Azar et al.\ call this form {\em error confinement}
\cite{Azar:2010}. More relaxed forms of fault-containment are known as
time-adaptive self-stabilization \cite{Kutten:2004}, scalable
self-stabilization \cite{Ghosh:2002}, strict stabilization
\cite{Nesterenko_02}, strong stabilization \cite{Dubois_12}, and
1-adaptive self-stabilization \cite{Beauquier_06}.

A large body of research focuses on fault-containing in the $1$-faulty
case. A configuration is called $k$-faulty, if in a legitimate
configuration exactly $k$ processes are hit by a fault. Several
metrics have been introduced to quantify the containment behavior in
the $1$-faulty case \cite{Ghosh:2007,Sven:2011}. A distributed
algorithm ${\cal A}$ has {\em contamination radius} $r$ if only nodes
within the $r$-hop neighborhood of the faulty node change their state
during recovery from a 1-faulty configuration. The {\em containment}
time of ${\cal A}$ denotes the worst-case number of rounds any
execution of ${\cal A}$ starting at a 1-faulty configuration needs to
reach a legitimate configuration. In technical terms this corresponds
to the {\em worst case time to recover} in case of a single fault. For
randomized algorithms the expected number of rounds to reach a
legitimate configuration corresponds to the {\em mean time to recover} (MTT).

The stabilization time is an obvious upper bound for the containment
time. In some cases this bound can be improved, for example when the
contamination radius is known. In the shared memory model an algorithm
with contamination radius $r$ and stabilization time $O(f(n))$
obviously has containment time $O(f(\Delta^r))$. There are only few
cases where the containment time is explicitly computed and in these
cases only asymptotic bounds are given. From a practical point of view
absolute bounds are more valuable.


The focus of this paper is on the analysis of the containment time of
randomized self-stabilizing algorithms in the message passing model
with respect to memory and message corruption. We show how Markov
chains can be used to find upper bounds for the containment time that
are lower than the above mentioned trivial bound $O(f(\Delta^r))$.
For a $\Delta+1$-coloring algorithm we analytically derive an
absolute bound for the expected containment time and show that the
variance is surprisingly bounded by a small constant independent of
the network's size. We believe that the presented techniques can also
be applied to other algorithms.

\section{Related Work}
There exist several techniques to analyze self-stabilizing algorithms:
potential functions, convergence stairs, Markov chains, etc. Markov
chains are particular useful for randomized algorithms
\cite{Duflot_01}. Their main drawback is that in order to set up the
transition matrix the adjacency matrix of the graph must be known.
This restricts the applicability of this method to {\em small} or
highly {\em symmetric} instances. DeVille and Mitra apply model
checking tools to Markov chains for cases of networks of small size
($n\le 7$) to determine the expected stabilization time
\cite{DeVille_09}. An example for highly symmetric networks are ring
topologies, see for example \cite{Fribourg_06,Yamashita_11}. Fribourg
et al. model randomized distributed algorithms as Markov chains using
the technique of coupling to compute upper bounds for the
stabilization times \cite{Fribourg_06}. Yamashita uses Markov chains
to model self-stabilizing probabilistic algorithms and to prove
stabilization \cite{Yamashita_11}. Mitton et al.\ consider a
randomized self-stabilizing $\Delta+1$-coloring algorithm and model
this algorithm in terms of urns/balls using a Markov chain to get a
bound for the stabilization time \cite{Mitton_06}. They evaluated the
Markov chain for networks up to 1000 nodes analytically and by
simulations. Crouzen et al.\ model faulty distributed algorithms as
Markov decision processes to incorporate the effects of random faults
when using a non-deterministic scheduler \cite{Crouzen_11}. They used
the PRISM model-checker to compute long-run average availabilities.
The above literature considered only the shared memory model.






\section{Bounding the Containment Time}
\label{sec:modell}
The containment time is a special case of the stabilization time. The
difference is that only executions starting from $1$-faulty
configurations are considered. Such configurations arise when a single
node $v$ is hit by a memory corruption or a single message sent by $v$
is corrupted. We do not consider corruptions of the code of an
algorithm. Denote by $R_v$ the subgraph of the communication
graph $G$ that is induced by the nodes that are engaged in the
recovery process from a $1$-faulty configuration triggered by a fault
at $v$. There are two situations in which it is apparently feasible to
obtain bounds for the containment time that are lower than the above
mentioned trivial bound of $O(f(|R_v|)$: Either the structure of $R_v$
is considerably simpler than that of $G$ or $R_v$'s size is smaller
than that of $G$.

\subsection{Shared Memory Model}
\label{sec:ssm}

First consider the shared memory model. If an algorithm has
contamination radius $r$ and no other fault occurs then this fault
will not spread beyond the $r$-hop neighborhood of the faulty node
$v$. In this case $R_v \subseteq G_v^r$, where $G_v^r$ is the subgraph
induced by $v$ and nodes $w$ with $dist(v,w)\le r$. The analysis of
the containment time is often simplified due to the fact that the
initial configuration is almost legitimate (i.e., only $v$ is not
legitimate).

As a first example consider the well known self-stabilizing Algorithm
$\A_{1}$ to compute a maximal independent set (MIS). 

\vspace*{1mm}
\begin{algorithm}[H]\footnotesize
  \If{$state = IN \wedge \exists w\in N(v)$ s.t. $w.state = IN$}{$state
  := OUT$}
  \If{$state = OUT \wedge \forall w\in N(v)$ $w.state = OUT$}{
      \If{random bit from {0,1} = 1}{$state
  := IN$}}
  \caption{Self-stabilizing algorithm $\A_{1}$ to compute a MIS.\label{alg:MIS}}
\end{algorithm}
\vspace*{1mm}


\begin{lemma}\label{lem:A1}
  Algorithm $\A_1$ has contamination radius two.
\end{lemma}
\begin{proof}
  Let $v$ be a node hit by a memory corruption. First suppose the
  state of $v$ changes from $IN$ to $OUT$. Let $u\in N(v)$ then
  $u.state=OUT$. If $u$ has an neighbor $w\not=v$ with $w.state=IN$
  then $u$ will not change its state during recovery. Otherwise, if
  all neighbors of $u$ except $v$ had state $OUT$ node $u$ may change
  state during recovery. But since these neighbors of $u$ have a
  neighbor with state $IN$ they will not change their state. Thus, in
  this case only the neighbors of $v$ may change state during
  recovery. 

  Next suppose that $v.state$ changes from $OUT$ to $IN$. Then $v$
  and those neigh\-bors of $v$ with state $IN$ can change to
  $OUT$. Then arguing as in the first case only nodes within distance
  two of $v$ may change their state during recovery.
\end{proof}

Graph $R_v$ can contain
any subgraph $H$ with $\Delta(G)$ nodes. For example let $G$ consist of
$H$ and an additional vertex $v$ connected to each node of $H$. A
legitimate configuration is given if the state of $v$ is $IN$ and all
other nodes have state $OUT$ (Fig.~\ref{fig:A1_cr} left). If $v$
changes its state to $OUT$ due to a fault then all nodes may change to
state $IN$ during the next round. A precise analysis of the
containment time depends extremely on the structure of $H$. Thus,
there is little hope for a bound below the trivial bound. Similar
arguments hold for the second 1-faulty configuration of $\A_{1}$ shown
on the right of Fig.~\ref{fig:A1_cr}.

\begin{figure}
\centering
\subfigure{\includegraphics[scale=0.875]{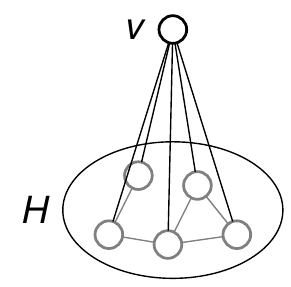}}\hspace*{12mm}
\subfigure{\includegraphics[scale=0.875]{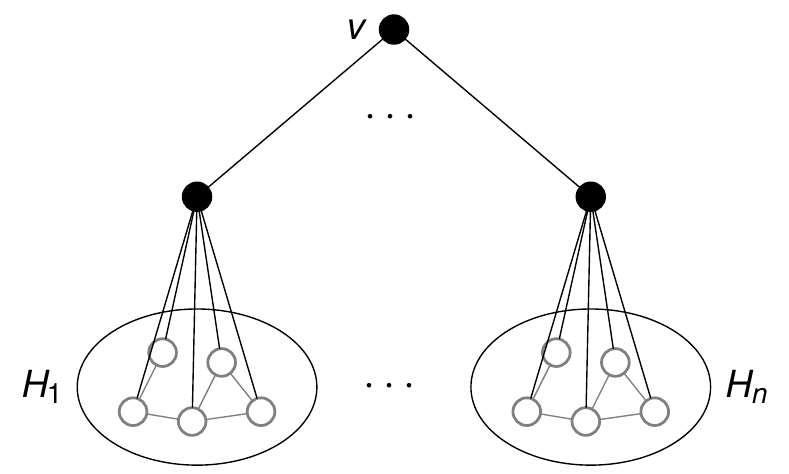}}
\caption{1-faulty configurations of $\A_{1}$ caused by a memory
  corruption at $v$. Nodes drawn in bold have state $IN$. The depicted
  graphs correspond to $R_v$.}\label{fig:A1_cr}
\end{figure}

Next we consider the problem of finding a $\Delta + 1$-coloring.
Almost all self-stabilizing algorithms for this problem follow the
same pattern. A node that realizes that it has chosen the same color
as one of its neighbors chooses a new color from a finite color
palette. This palette does not include the current colors of the
node's neighbors. To be executed under the synchronous scheduler these
algorithms are either randomized or use identifiers for symmetry
breaking. Variations of this idea are followed in
\cite{Dolev_97,Gradinariu:2000,Nesterenko_02,Mitton_06}. As an example
consider Algorithm $\A_{2}$ from \cite{Gradinariu:2000}. Due to its
choice of a new color from the palette $\A_{2}$ has contamination
radius at least $\Delta(G)$ (see Fig.~\ref{fig:A2_cr}).

\vspace*{2mm}
\begin{algorithm}[H]\footnotesize
    \If{$c \not= \max\left(\{0,\ldots,\Delta\}\backslash
        \{w.c\mid w\in N(v)\} \right)$}{
      \If{random bit from {0,1} = 1}{$c := \max\left(\{0,\ldots,\Delta\}\backslash
        \{w.c\mid w\in N(v)\} \right)$}
  }
  \caption{Self-stabilizing $\Delta + 1$-coloring algorithm $\A_{2}$ from
    \cite{Gradinariu:2000}.\label{alg:grad_color}}
\end{algorithm}
\vspace*{1mm}

\begin{figure}\label{fig:A2_cr}
\centering
\includegraphics[width=1\textwidth]{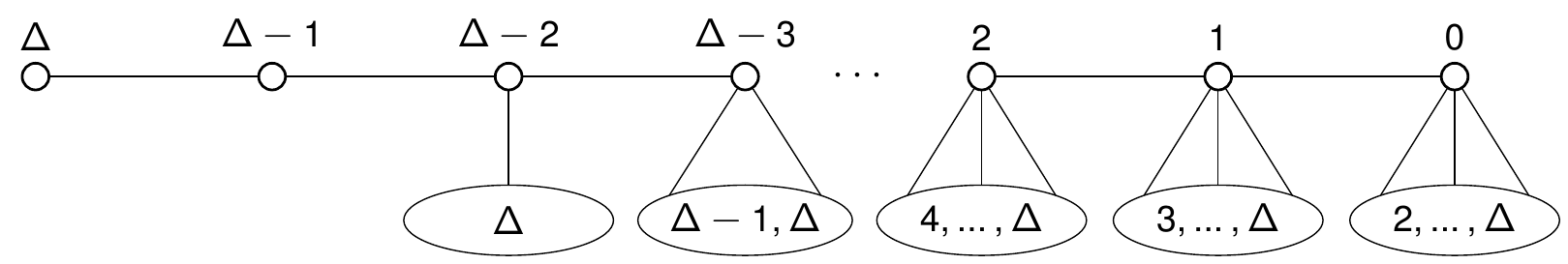}
\caption{Algorithm $\A_2$ has contamination radius $\Delta$: If the
  left-most node is hit by a fault and changes its color to
  $\Delta-1$, then all other nodes may change their color.}
\end{figure}

For some problems minor modifications of the algorithms can lead to
dramatic changes of the contamination radius. Algorithm $\A_{3}$ is a
slight modification of this algorithm. $\A_{3}$ has containment radius
$1$ (see Lemma~\ref{lem:A3}) and $R_v$ is a star graph with center
$v$. Note that neighbors of $v$ that change their color during
recovery form an independent set. This simple structure allows an
analysis of the containment time.

\vspace*{1mm}
\begin{algorithm}[H]\footnotesize
    \If{$\exists w\in N(v)$ s.t. $c = w.c$}{
      \If{random bit from {0,1} = 1}{$c :=$ choose $\{0,\ldots,\Delta\}\backslash\{w.c\mid w\in N(v)\}$}}
  \caption{Self-stabilizing $\Delta + 1$-coloring algorithm $\A_{3}$.}\label{alg:simp_color}
\end{algorithm}
\vspace*{1mm}

\begin{lemma}\label{lem:A3}
  Algorithm $\A_3$ has contamination radius one.
\end{lemma}
\begin{proof}
  Let $v$ be a node hit by a memory corruption changing its color to a
  color $c$ already chosen by at least one neighbor of $v$. Let
  $N_{conf}=\{w\in N(v)\mid w.c = c\}$. In the next round the nodes in
  $N_{conf} \cup \{v\}$ will get a chance to choose a new color. The
  choices will only lead to conflicts between $v$ and others nodes in
  $N_{conf}$. Thus, the fault will not spread beyond the set
  $N_{conf}$. With a positive probability the set $N_{conf}$ will
  contain fewer nodes in each round.
\end{proof}

\subsection{Message Passing Model}
\label{sec:mpm}
In the message passing model the situation is different for two
reasons. First of all, a $1$-faulty configuration is also given when a
single message sent by a node $v$ is corrupted. Secondly, this may
cause neighbors of $v$ to send messages they would not send in a
legitimate configuration. Even so the state of nodes with distance
greater than $r$ to $v$ does not change, these nodes may be forced to
send messages. Thus, in general the analysis of the containment time
cannot be performed by considering $G_v^r$ only. This is only possible
in cases when a fault at $v$ does not force nodes at distance $r+1$ to
send messages they would not send had the fault not occurred.

\section{Computing the Expected Containment Time}
\label{sec:mc}
A randomized synchronous self-stabilizing algorithm $\AG$ can be
regarded as a transition systems. Denote by $\Sigma$ the set of all
configurations. In each round the current configuration $c\in \Sigma$
is transformed into a new configuration $\AG(c)\in \Sigma$. This
process is described by the transition matrix $P$ where $p_{ij}$ gives
the probability to move from configuration $c_i$ to $c_j$ in one
round, i.e., $\AG(c_i)=c_j$.

To compute the containment time one must consider all executions
starting from a $1$-faulty configuration $c$. Let $X$ be the random
variable that denotes the number of rounds until the system has
reached a legitimate configuration when starting in $c$. The expected
containment time equals the expected value $E[X]$. In some
cases it is possible to compute $E[X]$ directly according to the
definition. But in most cases this will be impossible.

To reduce the complexity it is often helpful to partition $\Sigma$
into subsets $S_0,\ldots, S_l$ and consider these as the states of a
Markov chain. The subsets must have the property that for each pair of
subsets $S_i,S_j$ the probability of a configuration $c\in S_i$ to be
transformed in one round into a configuration of $S_j$ is independent
of the choice of $c \in S_i$. This probability is then interpreted as
the transition probability from $S_i$ to $S_j$. This way the
complexity of the analysis can often be reduced dramatically.

A state $c_i$ of a Markov chain is called {\em absorbing} if
$p_{ii}=1$ and $p_{ij}=0$ for $i\not= j$. For self-stabilizing
algorithms, the set of all legitimate configurations $\cal L$ is an
absorbing state, in fact it is the unique absorbing state in this
case. The number of rounds to reach a legitimate configuration starting
from a given configuration $c_i\in S_i$ equals the number of steps
before being absorbed in $\cal L$ when starting in state $S_i$. This
equivalence allows to use techniques from Markow chains to compute the
stabilization time and thus, the containment time.

To compute the containment time we must consider all executions
starting from a fixed $1$-faulty configuration. Let $S_0$ consist of a
single $1$-faulty configuration and $S_l$ be the set of all legitimate
configurations. It suffices to compute the expected number of rounds
to reach $S_l$ from $S_0$ and then take the maximum for all $1$-faulty
configurations.

\subsubsection{Example}
\label{sec:example1}
As an example consider algorithm $\A_{3}$ as described above. Let
$c$ be a legitimate configuration and $v$ be a node that changes its
color due to a memory fault. This causes a conflict with all $d$
neighbors of $v$ that had chosen this color. During the execution of
$\A_{3}$ only nodes contained in $R_v$ (a star graph) change their
state. Furthermore, once a neighbor has chosen a color different from
$v$ then the choice will be forever (at least until the next transient
fault).

Let $d$ be the number of neighbors of $v$ that have the same color as
$v$ after the fault. Denote by $S_j$ the set of all configurations
reachable from $c$ where exactly $d-j$ neighbors of $v$ have the same
color as $v$. Then $S_0=\{c\}$ and $S_d$ consists of all legitimate
configurations. Let $c_i\in S_i$. Then $\A_{3}(c_i)\not\in S_j$
for all $j < i$. Unfortunately, the probability of a configuration
$c_i\in S_i$ to be transformed in one round into a configuration of
$S_j$ for $j>i$ is not independent of the choice of $c_i \in S_i$. But
it is possible to resolve this issue as shown below.

\subsection{Absorbing Markov Chains}
\label{sec:absorb-mark-chains}
Let $S_0,\ldots,S_{d}$ be nonempty subsets of $\Sigma$ and $P$ a
stochastic matrix such that $P(\AG(c)\in S_j)=p_{ij}$ for all $i,j\in
\{0,1,\ldots,d\}$ and all $c\in S_i$. Furthermore, let $S_{d}$ be
the single absorbing state. Let $Q$ be the matrix obtained from $P$ by
removing the last row and the last column. $Q$ describes the
probability of transitioning from some transient state to another. The
following properties about absorbing Markov chains are well known and
can be found in Theorem 3.3.5 of \cite{Kemmeny:1976}.
Denote the $d \times d$ identity matrix by $E_d$. Then \[N = (E_d -
Q)^{-1}\] is the fundamental matrix of the Markov chain. The expected
number of steps before being absorbed by $S_{d}$ when starting from
 $S_i$ is the $i$-th entry of vector \[a = NI_d\] where $I_d$ is
a length-$d$ column vector whose entries are all 1. The variance of
these numbers of steps is given by the entries of \[(2N-E_d)a -
a_{sq}\] where $a_{sq}$ is derived from $a$ by squaring each entry.

\subsubsection{Example}
\label{sec:example2}
To apply the results of the last section to algorithm $\A_{3}$ the
following adjustments are made. For $i< j$ let $p_{ij}$ be a constant
such that $P(\A_{3}(c_i)\in S_j) \ge p_{ij}$ for all $c_i\in S_i$.
Furthermore, let $p_{ij} =0$ for $j < i$ and \[p_{ii} = 1-
\sum_{k=j+1}^d p_{ij}\] for $i=0,\ldots, d$. Then matrix $(p_{ij})$ is
a stochastic matrix with $p_{dd}=1$. Furthermore, the expected number
of steps before being absorbed by state $d$ when starting from state
$i$ is an upper bound for the expected number of rounds before being
absorbed by state $S_{d}$ when starting from state $S_i$. Thus, the
results from the last section can be used to find an upper bound for
the expected containment time of algorithm $\A_{3}$.


These techniques are applied to the self-stabilizing coloring
algorithm $\A_{col}$.

\section{Algorithm $\A_{col}$}
Computing a $\Delta + 1$-coloring in expected $O(\log n)$ rounds with
a randomized algorithm is long known \cite{Luby:1986,Johansson:1999}.
The $(\Delta +1)$-coloring algorithm $\A_{col}$ analyzed in this work
follows the pattern sketched in section~\ref{sec:ssm}. It is 
derived from a non-self-stabilizing algorithm contained in
\cite{Barenboim:2013} (Algorithm 19). The presented techniques can
also be applied to other randomized coloring algorithms such as those
described in \cite{Dolev_97,Gradinariu:2000,Nesterenko_02,Mitton_06}.
The main difference is that $\A_{col}$ assumes the message passing
model, more precisely the synchronous ${\cal CONGEST}$ model as
defined in \cite{Peleg:2000}. $\A_{col}$ stabilizes after
$O(\log n)$ rounds with high probability whereas the above cited
algorithms require a linear number of rounds.

At the beginning of each round of $\A_{col}$ each node broadcasts its
current color to its neighbors. Based on the information received from
its neighbors a node decides either to keep its color (final choice),
to choose a new color or no color (value $\bot$). In particular with
equal probability a node $v$ draws uniformly at random a color from
the set $\{0, 1, \dots,\delta(v)\}\backslash tabu$ or indicates that
it made no choice (see Function {\tt randomColor}). Here, $tabu$ is
the set of colors of neighbors of $v$ that already made their final
choice.

\vspace*{3mm}
\begin{algorithm}[H]\footnotesize
  Color \randomColor{Node v, Set<Color> tabu}\block{
    \lIf{random bit from {0,1} = 1}{
      \Return{$\bot$}
    }
    \Return random color from $\{$0,1,$\ldots, \delta(v)\}\backslash$tabu\;
  }\vspace*{1mm}
\end{algorithm}
\vspace*{1mm}

In the algorithm of \cite{Barenboim:2013} a node maintains a list with
the colors of those neighbors that made their final choice. A fault
changing the content of this list is difficult to contain.
Furthermore, in order to notice a memory corruption at a neighbor,
each node must continuously send its state to all its neighbors and
cannot stop to do so. This is the price of self-stabilization and well
known \cite{Dolev:2000}. These considerations lead to the design of
Algorithm $\A_{col}$. Each node only maintains the chosen color
(variable {\tt c}) and whether its choice is final (variable {\tt
  final}). In every round a node sends these two variables to all
neighbors. To improve fault containment a node's final choice of a
color is only withdrawn if it coincides with the final choice of a
neighbor. To achieve a $\Delta +1$-coloring a node makes a new
choice if its color is larger than its degree. Note that this
situation can only originate from a fault.

\vspace*{2mm}
\begin{algorithm}[H]\footnotesize
  Set<Color> tabu := $\emptyset$,  occupied := $\emptyset$\;
  broadcast(c, final) to all neighbors w $\in$ N(v)\;
  \For{all neighbors w $\in$ N(v)}{
     receive(c$_w$, final$_w$) from node $w$\;
     \If{$c_w\not=\bot$}{
        occupied := occupied $\cup~\{$c$_w\}$\;
        \lIf{final$_w$}{
           tabu := tabu $\cup~\{$c$_w\}$
        }
     }
  }
  \eIf{c = $\bot \vee c > \delta(v)$}{
      final := false\;
  }{
    \eIf{final}{
      \lIf{$c \in$ tabu}{
        final := false
      }
    }{
      \lIf{$c \not\in$ occupied}{
        final := true
      }
    }
  }
  \lIf{final = false}{
    c := \randomColor{v, tabu}
  }\vspace*{3mm}
  \caption{$\A_{col}$ as executed by a node $v$.}\label{alg:col}
\end{algorithm}
\vspace*{2mm}

Theorem~\ref{theo:stab_time} states the correctness and the stabilization
time of \AC. The algorithm works correctly for any initial setting
of the variables. Note that if $v.final=true$ one round after a
transient fault or the initial start $v$ will not change its color
until the next fault. With this observation the theorem can be proved
along the same lines as Lemma~10.3 in \cite{Barenboim:2013}.

\begin{theorem}\label{theo:stab_time}
  Algorithm \AC is self-stabilizing and computes a $\Delta + 1$-coloring
  within $O(\log n)$ time with high probability (i.e. with probability
  at least $1-n^c$ for any $c \ge 1$). \AC has contamination radius $1$.
\end{theorem}

First we consider {\em error-free} executions, i.e.\ executions during
which no memory nor message corruptions occurs. Note that \AC must
work correctly for any initial setting of the variables. A
configuration is called a {\em legal coloring} if the values of
variable $c$ form a $\Delta +1$-coloring of the graph. It is called
{\em legitimate} if is a legal coloring and $v.final=true$ for each
node $v$. A node $v$ {\em pauses} in round $r$ if it does not change
the value of $v.c$ or $v.final$ in round $r$. A node $v$ {\em
  terminates} in round $r$ if it pauses in round $r$ and all following
rounds.

\begin{lemma}\label{lem:silent}
  Let $e$ be an error-free execution and let $v\in V$. If $v$ pauses
  in round $r$ of $e$ then $v.final=true$ and $v$ terminates in round
  $r$.
\end{lemma}
\begin{proof}
  The only constellation at the beginning of a round in which $v$
  pauses is $v.final= true$, $c\in\{0,1,\ldots,\delta(v)\}$, and $c
  \not\in v.tabu$. The latter condition implies that each neighbor
  $w$ of $v$ with $w.c=v.c$ at the beginning of round $r$ has sent
  $w.final = false$ to $v$. Since $v$ sent $(v.c,true)$ to each
  neighbor, no node $w\in N(v)$ will at the end of round $r$ have
  $w.c=v.c$ and $w.final=true$. This implies, that in the following
  round still $v.c\not\in v.tabu$ holds. Thus, $v$ also pauses in
  round $r+1$. This proves that $v$ terminates in round $r$.
\end{proof}

\begin{lemma}\label{lem:init}
  Let $r$ be a round of an error-free execution and let $v\in V$. If
  $v.final= true$ at the beginning and $v.final=false$ at the end of
  round $r$ then $r=1$.
\end{lemma}
\begin{proof}
  Denote the value of $v.c$ at the beginning of round $r$ by $c_r$. In
  order for $v$ to set $v.final$ to $false$ one of the following three
  conditions must be met at the beginning of round $r$: 
  \begin{enumerate}
  \item $c_r > \delta(v)$,
  \item $c_r = \bot$, or
  \item $v$ has a neighbor $w$ with $w.final=true$ and $w.c=c_r$.
  \end{enumerate}
  The first condition can only be true in round $1$. Suppose that
  $c_r=\bot$ and $v.final=true$ at the beginning of round $r$ with
  $r>1$. If during round $r-1$ the value of $v.final$ was set to
  $true$ then $v.c$ could not be $\bot$. Hence, at the beginning of
  round $r-1$ already $v.final=true$. But then $v.c$ was not changed
  in round $r-1$, hence $v.c=\bot$ at the beginning of round $r-1$,
  i.e. $v$ paused in round $r-1$ but did not terminate. This is a
  contradiction with Lemma~\ref{lem:silent}. Finally assume the last
  condition. Then $v$ and $w$ cannot have changed their value of
  variable $c$ in round $r-1$, because then variable $final$ could not
  have value $true$ at the beginning of round $r$. Thus, $v$ sent
  $(c_r, true)$ in round $r-1$. Hence, if $w.c=c_r$ at that time $w$
  would have changed $w.final$ to $false$, contradiction.
\end{proof}

\begin{lemma}\label{lem:terminate}
  A node setting $final$ to $true$ in round $r$ terminates in round
  $r+1$.
\end{lemma}
\begin{proof}
  A node $v$ that sets $v.final$ to $true$ satisfies
  $v.c\in\{0,1,\ldots,\delta(v)\}$ and all $w\in N(v)$ have $w.c\not=v.c$
  at the beginning of round $r$. Also $v$ does not change its color
  during round $r$. Thus, no  $w\in N(v)$ will change its
  color to $v.c$ during round $r$. Thus, at the beginning or the next
  round $v.final=true$ and $v.c \not\in v.tabu$. This yields that $v$
  pauses in round $r+1$. The result follows from
  Lemma~\ref{lem:silent}.
\end{proof}

\begin{lemma}\label{lem:vlaid_color}
  If all nodes have terminated  the configuration is legitimate.
\end{lemma}
\begin{proof}
  Let $r$ be a round in which all nodes pause. By
  Lemma~\ref{lem:silent} all $v\in V$ satisfy $v.final=true$ in this
  round. Furthermore, since no node changes variable $c$ in round $r$,
  $v.c\not\in v.tabu$ for each $v\in V$. Thus, $v.c\not=w.c$ for each
  $w\in N(v)$ and therefore variable $c$ constitutes a valid coloring.
  Finally, note that because of $v.c\in \{0,1,\ldots,\delta(v)\}$ at most
  $\Delta +1$ colors are used.
\end{proof}

\begin{proof}[Theorem~\ref{theo:stab_time}]
  According to Lemma~\ref{lem:vlaid_color} it suffices to prove that
  all nodes terminate within $O(\log n)$ time with high probability.
  Let $v\in V$.  Lemma~\ref{lem:terminate} implies that
  the probability that $v$ terminates in round $r>1$ is equal to the
  probability that $v$ sets its variable $v.final$ to $true$ in round
  $r-1$. This is the probability that $v$ selects a color different
  from $\bot$ and from the selections of all neighbors that chose a
  value different from $\bot$ in round $r-2$. Suppose that indeed
  $v.c\not=\bot$ at the end of round $r-2$. Then $v.c\not\in v.tabu$.
  The probability that a given neighbor $u$ of $v$ selects the same
  color $u.c = v.c$ in this round is at
  most $\frac{1}{2(\delta(v)+1-|v.tabu|)}$. This is because the
  probability that $u$ selects a color different from $\bot$ is $1/2$,
  and $v$ has $\delta(v) + 1 - |v.tabu|$ different colors to select from.
  Since $r>1$ all nodes in $v.tabu$ have $final = true$ and will never
  change this value. Thus, at most $\delta(v) - |tabu|$ neighbors select a
  new color.
By the union bound, the probability that
  $v$ selects the same color as a neighbor
  is at most \[\frac{\delta(v)-|v.tabu|}{2(\delta(v)+1-|v.tabu|)} <
  \frac{1}{2}.\] Thus, if $v$ selects a color $v.c\not=\bot$, it is
  distinct from the colors of its neighbors with probability at least
  $1/2$. It holds that $v.c\not=\bot$ with probability $1/2$. Hence, $v$
  terminates with probability at least $1/4$.

  The probability that a specific node $v$ doesn't terminate within
  $r$ rounds is at most $(3/4)^r$ . By the union bound, the
  probability that there exists a vertex $v\in V$ that does not
  terminate within $r$ rounds is at most $n (3/4)^r$. Hence, ${\cal
    A}_{col}$ terminates after $r=(c + 1) 4\log n$ rounds, with
  probability at least $1 - n (3/4)^r \ge 1 - 1/n^c$ (note that $\log 4/3
  > 1/4$).
\end{proof}

\section{Fault Containment of Algorithm \AC}
In this section the fault containment behavior of \AC is analyzed. In
particular we consider a legitimate configuration in which a single
transient error occurs. Two types of transient errors are considered:
\begin{enumerate}
\item Memory corruption at node $v$, i.e., the value of at
  least one of the two variables of $v$ is corrupted.
\item A broadcast message sent by $v$ is corrupted. Note that the
  alternative implementation of using $\delta(v)$ unicast messages
  instead a single broadcast has very good fault containment behavior
  but is much slower.
\end{enumerate}
The {\em independent degree} $\delta_i(v)$ of a node $v$ is the size
of a maximum independent set of $N(v)$. Let $\Delta_i(G) =
\max\{\delta_i(v) \mid v\in V\}$.

\subsection{Message Corruption}
First consider the case that a single broadcast message sent by $v$ is
corrupted, i.e. the message contains a color $c_f$ different from
$v.c$ or the value $false$ for variable $final$. Since $w.final=true$
for all $w\in N(v)$ the message $(c_f,false)$ has no effect on any
$w\in N(v)$ regardless of the value of $c_f$. Thus, this corrupted
message has no effect at all.

Next consider the case that $v$ broadcasts the message $(c_f,true)$ with
$c_f\not=v.c$. Let $N_{conf}(v) = \{w\in N(v)\mid w.c=c_f\}$. The
nodes in $N_{conf}(v)$ form an independent set, because they all have
the same color. Thus $|N_{conf}(v)|\le \delta_i(v)$.

\begin{lemma}\label{lem:contain}
  The contamination radius after a single corruption of a broadcast
  message sent by node $v$ is $1$, in particular neither $v$ nor a
  node outside $N_{conf}(v)$ will change its state. At most
  $\delta_i(v)$ nodes change their state during recovery.
\end{lemma}
\begin{proof}
  Let $u\in V\backslash N[v]$. This node continues to send
  $(u.c,true)$ after the fault. Thus, a neighbor of $u$ that changes
  its color will not change its color to $u.c$. This yields that no
  neighbor of $u$ will ever send a message with $u.c$ as the first
  parameter. This is also true in case $u\in N(v)\backslash
  N_{conf}(v)$. Hence, no node outside $N_{conf}(v) \cup\{v\}$ will
  change its state, i.e.\ the contamination radius is $1$.

  Next consider the node $v$ itself. Let $w\in N_{conf}(v)$. When the
  faulty message is received by $w$ it sets $w.final$ to false. Before
  the faulty message was sent no neighbor of $v$ had the same color as
  $v$. Thus, in the worst case a node $w\in N_{conf}(v)$ will choose
  $v.c$ as its new color and send $(v.c, false)$ to all neighbors.
  Since $v.final=true$ this will not force $v$ to change its state.
  Thus, $v$ keeps broadcasting $(v.c,true)$ and therefore no neighbor
  $w$ of $v$ will ever reach the state $w.c=v.c$ and $w.final=true$.
  Hence $v$ will never change its state.
\end{proof}

Theorem~\ref{theo:stab_time} implies that the containment time of
this fault is $O(\log \delta_i(v))$ with high probability. The
following lemma gives a more precise bound.

\begin{lemma}\label{lem:cor_mess}
  The expected value for the containment time after a corruption of a
  message broadcasted by node $v$ is at most $\frac{1}{\ln 2}
  H_{\delta_i(v)}+1/2$ rounds ($H_i$ denotes the $i^{th}$ harmonic
  number) with a variance of at most \[ \frac{1}{\ln^2 2}\sum_{i=1}^{\delta_i(v)}
  \frac{1}{i^2} +\frac{1}{4} \le \frac{\pi^2}{6\ln^2
    2}+\frac{1}{4}\approx 3.6737.\]
  \end{lemma}
\begin{proof}
  After receiving message $(c_f,true)$ all nodes $w\in N_{conf}(v)$
  set $w.final$ to $false$ and with equal probability $w.c$ to $\bot$
  or to a random color $c_w\in \{0,1,\ldots, \delta(w)\}\backslash
  w.tabu$. Note that $|w.tabu|\le \delta(w)$ because $w.tabu = \{u.c \mid
  u\in N(w) \backslash v\} \cup \{c_f\}$. If $w$ chooses a color
  different from $\bot$ then this color is different from the colors
  of all of $w$'s neighbors. Also in this case $w$ will terminate
  after the following round because then it will set $final$ to
  $true$. Thus, after one round $w$ has chosen a color that is
  different from the colors of all neighbors with probability at least
  $1/2$. Furthermore, this color will not change again. After one
  additional round $w$ reaches a legitimate state.

  Let the random variable $X_d$ with $d=|N_{conf}(v)|$ denote the
  number of rounds until the system has reached a legal coloring.
  For $w\in N_{conf}(v)$ let $Y_w$ be the random variable denoting the
  number of rounds until $w$ has a legal coloring. By
  Lemma~\ref{lem:contain} $X_d = \max \{Y_w\mid w\in N_{conf}(v)\}$.
  For $i\ge 1$ let $G(i) = P\{X_d\le i\} = P\{\max \{Y_w\mid w\in
  N_{conf}(v)\}\le i\}$. Since the random variables $Y_w$'s are
  independent $G(i) = (P\{X\le i\})^{|N_{conf}(v)|}+1$ where $X$ is a
  geometric random variable with $p=0.5$. Thus, \vspace*{-2mm}
  \[G(i) = \left(\sum_{j=1}^i p(1-p)^{j-1}\right)^d\] and $G(0) = 0$
  with $d=|N_{conf}(v)|$. Then $E[X_d] = \sum_{i=1}^\infty ig(i)$ with
  probability function $g(i) = P\{X_d=i\}$. Let $q = 1-p$. Now for $i\ge1$
\[g(i) = G(i) - G(i-1)=(1-q^i)^d -(1-q^{i-1})^d = \]\[=\sum_{j=0}^d \binom{d}{j}
(-1)^{j+1}(1-q^j)q^{j(i-1)}  =\sum_{j=1}^d \binom{d}{j}
\frac{(-1)^{j+1}(1-q^j)}{q^j}(q^j)^i.\]

\noindent
This implies

\begin{align*}
E[X_d] &=   \sum_{j=1}^d\binom{d}{j}
\frac{(-1)^{j+1}(1-q^j)}{q^j} \sum_{i=1}^\infty i(q^{j})^i = \sum_{j=1}^d\binom{d}{j}
\frac{(-1)^{j+1}}{(1-q^j)}\\
&=\sum_{j=1}^d\binom{d}{j}
(-1)^{j+1}\sum_{l=0}^\infty (q^j)^l= \sum_{l=0}^\infty \sum_{j=1}^d\binom{d}{j}
(-1)^{j+1}(q^l)^j\\
&=\sum_{l=0}^\infty \left(1 + \sum_{j=0}^d\binom{d}{j}
(-1)^{j+1}(q^l)^j\right) = \sum_{l=0}^\infty (1-(1-q^l)^d)
\end{align*}
The result follows from Lemma~\ref{lem:harmo}. The derivation of the
formula for the variance can be found in Lemma~\ref{lem:var}.
\end{proof}

\begin{lemma}\label{lem:harmo}
  For fixed  $0<q<1$ and fixed $d\ge 1$
  \[\sum_{l=0}^\infty (1-(1-q^l)^d)\in [-\frac{1}{\ln q}
  H_d,-\frac{1}{\ln q} H_d+1] \text{  ~and~  } \sum_{l=0}^\infty
  (1-(1-q^l)^d) \approx -\frac{1}{\ln q} H_d +\frac{1}{2}.\]
\end{lemma}
\begin{proof}
  The function $f(x) = 1 - (1-q^x)^d$ is for fixed values of $d$
  decreasing for $x\ge0$. Furthermore, $f(0)=1$. Hence,
  \[\sum_{l=0}^\infty (1-(1-q^l)^d) \ge \int_0^\infty f(x) dx \ge \sum_{l=0}^\infty (1-(1-q^l)^d)-1.\] Using the
  substitution $u = 1-q^x$ the integral becomes

  \[-\frac{1}{\ln q} \int_0^\infty \frac{1-u^d}{1-u} du =
  -\frac{1}{\ln q} \int_0^1 \sum_{i=0}^{d-1} u^idu = -\frac{1}{\ln
    q} \sum_{i=1}^{d} \frac{1}{i} = -\frac{1}{\ln q} H_d.\]

Approximating $\int_i^{i+1} f(x)d x$ with $(f(i)+f(i+1))/2$
yields \[\sum_{l=0}^\infty (1-(1-q^l)^d) \approx \int_0^\infty f(x) dx
+ \frac{f(0)}{2} = -\frac{1}{\ln q} H_d +\frac{1}{2}\]
\end{proof}

\begin{lemma}\label{lem:var}
For $d>0$ the variance of the containment time is at most
 \[ Var[X_d]= \frac{1}{\ln^2 2}\sum_{i=1}^d \frac{1}{i^2} +\frac{1}{4} \le
  \frac{\pi^2}{6\ln^2 2}+\frac{1}{4}\approx 3.6737.\]
\end{lemma}
\begin{proof}\vspace*{-6mm}
  \begin{align*}
Var[X_d] &= E[X_d^2] - E[X_d]^2= \sum_{i=1}^\infty i^2 g(i) - E[X_d]^2
\end{align*}
By Lemma~\ref{lem:comp1}
\begin{align*}  
    \sum_{i=1}^\infty i^2 g(i) &=  \sum _{l=1}^\infty (2l+1) (1-(1-q^l)^d)
=  2\sum _{l=1}^\infty l (1-(1-q^l)^d) + E[X_d]
  \end{align*}
Now Lemma~\ref{lem:comp2} yields
\begin{align*}
Var[X_d] &\approx \frac{2}{\ln^2 2}\sum_{i=1}^d \frac{H_i}{i} + E[X_d] - E[X_d]^2\\
&\approx \frac{2}{\ln^2 2}\sum_{i=1}^d \frac{H_i}{i} + \frac{H_d}{\ln 2} + \frac{1}{2} - \left( \frac{H_d}{\ln 2} + \frac{1}{2}\right)^2 & \text{Lemma~\ref{lem:cor_mess}}\\
& = \frac{1}{\ln^2 2}\left(2\sum_{i=1}^d \frac{H_i}{i} - H_d^2\right) +\frac{1}{4}
\end{align*}
\begin{align*}
  2\sum_{i=1}^d \frac{H_i}{i} -H_d^2 &= 2\sum_{i=1}^d \sum_{j=1}^i \frac{1}{ij} -  (1 + \frac{1}{2} + \cdots + \frac{1}{d})^2\\
&= 2\sum_{i=1}^d \frac{1}{i^2} + 2\sum_{i=1}^d \sum_{j=1}^{i-1} \frac{1}{ij}  - 2\sum_{i=1}^d \sum_{j=1}^{i-1} \frac{1}{ij}- \sum_{i=1}^d \frac{1}{d^2}= \sum_{i=1}^d \frac{1}{i^2}=\frac{\pi^2}{6}
\end{align*}
\end{proof}

\begin{lemma}\label{lem:comp1}
 Let $d>0$, $q\in (0,1)$ and $g(i)=\sum\limits_{j=1}^d \binom{d}{j}
\frac{(-1)^{j+1}(1-q^j)}{q^j}(q^j)^i$. Then
  \[\sum_{i=1}^\infty i^2 g(i) = \sum _{l=1}^\infty (2l+1) (1-(1-q^l)^d)\]
\end{lemma}
\begin{proof}
  \begin{align*}
    \sum_{i=1}^\infty i^2 g(i) &= \sum_{j=1}^d\binom{d}{j}
    \frac{(-1)^{j+1}(1-q^j)}{q^j} \sum_{i=1}^\infty i^2(q^{j})^i& \\
    & = \sum_{j=1}^d\binom{d}{j}
    \frac{(-1)^{j+1}(1-q^j)}{q^j} \left(\frac{2q^{2j}}{(1-q^j)^3} + \frac{q^{j}}{(1-q^j)^2}\right)& \\
    & = \sum_{j=1}^d\binom{d}{j} (-1)^{j+1} \left( \frac{2q^{j}}{(1-q^j)^2} + \frac{1}{1-q^j}\right) \\
&= \sum_{j=1}^d\binom{d}{j} (-1)^{j+1}  \sum_{l=0}^\infty (2l+1)(q^j)^l &\\
&= \sum _{l=1}^\infty (2l+1) \sum_{j=1}^d\binom{d}{j} (-1)^{j+1} (q^l)^j \\
&= \sum _{l=1}^\infty (2l+1) (1-(1-q^l)^d)
  \end{align*}
 For the first equation we refer to the proof of
 Lemma~\ref{lem:cor_mess}. The second equality makes use
 of \[\sum_{i=1}^\infty i^2x^i = \frac{2x^{2}}{(1-x)^3} +
 \frac{x}{(1-x)^2}\] and the fourth equality uses the following two
 identities \[\sum _{l=0}^\infty x^l=\frac{1}{(1-x)}  \text{ and } \sum _{l=0}^\infty lx^l=\frac{x}{(1-x)^2}.\]
\end{proof}

\begin{lemma}\label{lem:comp2} Let $d>0$ and $q\in (0,1)$ then
  \[\sum _{l=1}^\infty l (1-(1-q^l)^d) \le \frac{1}{(\ln q)^2}\sum_{i=1}^{d} \frac{H_i}{i} \]
\end{lemma}
\begin{proof}
  We will approximate $\sum _{l=1}^\infty l (1-(1-q^l)^d)$ with
  $\int_0^\infty x(1-(1-q^x)^d)dx$. Note that $x(1-(1-q^x)^d)$ has a
  single local maximum in the interval $[0,\infty)$. 
 If the local maximum is with the
   interval $[y,y+1]$ within $y\in \mathbb{N}$ then the error is
 \[ \int_y^{y+1} x(1-(1-q^x)^d)dx.\] This leads to a
  small overestimation of the sum as Fig.~\ref{fig:Var_app} shows.
  \begin{align*}
    \sum _{l=1}^\infty l (1-(1-q^l)^d) & \le
    \int_0^\infty x(1-(1-q^x)^d)dx \\&= \frac{-1}{(\ln q)^2}\int_0^1 \ln(1-u)\frac{(1-u^d)}{1-u}du &\\
    & = \frac{-1}{(\ln q)^2}\sum_{i=0}^{d-1} \int_0^1 \ln(1-u)u^idu\\
    & = \frac{1}{(\ln q)^2}\sum_{i=1}^{d} \frac{H_i}{i} &
  \end{align*}
  The first equation uses the substitution $u = 1-q^x$. The final
  result is based on the following identity \[\int_0^1 \ln(1-u)u^idu=
  \frac{-H_{i+1}}{i+1}.\] 
\end{proof}

We evaluated the results of Lemma~\ref{lem:cor_mess} by modeling the
behavior of this fault situation as a Markov chain and computed
$E[X_d]$ and $Var[X_d]$ using Theorem 3.3.5 from
\cite{Kemmeny:1976}. These computations showed that $\frac{1}{\ln 2}
H_{d}+1/2$ matches very well with $E[X_d]$ and that $E[X_d]\approx
2\log d$ (see Fig.~\ref{fig:E_log}). Furthermore, the gap between
$Var[X_d]$ and the bound given in Lemma~\ref{lem:cor_mess} is less
than $0.2$ (see Fig.~\ref{fig:Var_app}).

\begin{figure}[H]
  \centering \subfigure[Comparison of comput\-ed value of $E{[X_d]}$
  with ${\log d}$
  (Lemma~\ref{lem:cor_mess}).]{\label{fig:E_log}\includegraphics[width=0.32\textwidth]{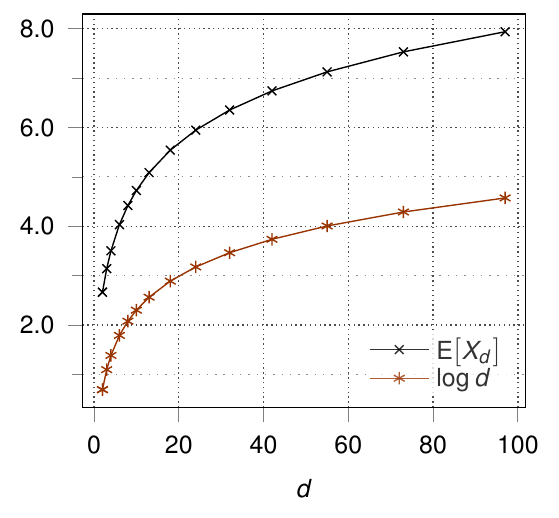}}\hfill
  \subfigure[Comparison of comput\-ed value of $Var{[X_d]}$ with
  approximation
  (Lemma~\ref{lem:cor_mess}).]{\label{fig:Var_app}\includegraphics[width=0.32\textwidth]{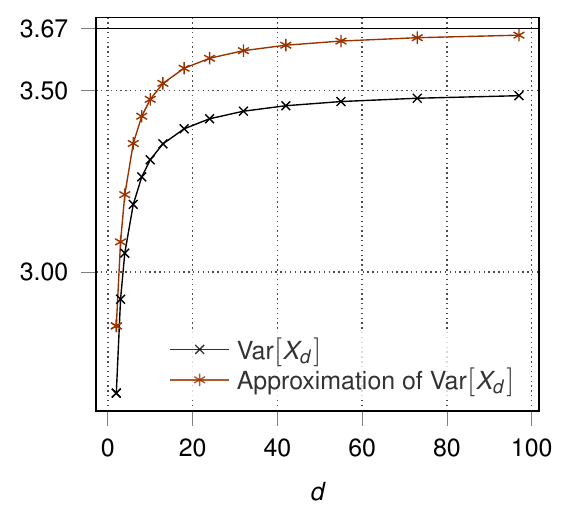}}\hfill
  \subfigure[$E{[A_d]}$ and $Var{[A_d]}$ from
  Lemma~\ref{lem:bound5}.]{\label{fig:E_Var}\includegraphics[width=0.30\textwidth]{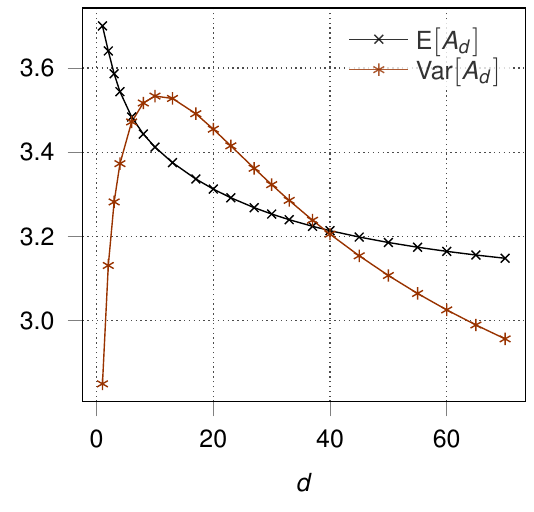}}
  \caption{Comparisons of computed with approximated values from
    Lemma~\ref{lem:cor_mess} and \ref{lem:bound5}.}
\end{figure}

\subsection{Memory Corruption}
This section considers the case that the memory of a single node $v$
is corrupted. First consider the case that the fault causes variable
$v.final$ to change to $false$. If $v.c$ does not change, then a
legitimate configuration is reached after one round. So assume $v.c$
also changes. Then the fault will not affect other nodes. This is
because no $w\in N(v)$ will change its value of $w.c$ because
$w.final=true$ and $v.final=false$. Thus, with probability at least
$1/2$ node $v$ will choose in the next round a color different form
the colors of all neighbors and terminate one round later. Similar to
$X_d$ let random variable $Z_d$ denote the number of rounds until a
legal coloring is reached ($d=|N_{conf}(v)|$). It is easy to verify
that $E[Z_d]=3$ in this case.

The more interesting case is that only variable $v.c$ is affected
(i.e. $v.final$ remains $true$). Let $c_f$ the corrupted value of
$v.c$ and suppose that $N_{conf}(v) = \{w\in N(v)\mid w.c=c_f\}\not=
\emptyset$. A node not contained in $S= N_{conf}(v) \cup \{v\}$ will
not change its state (c.f.\ Lemma~\ref{lem:contain}). Thus, the
contamination radius is $1$ and at most $\delta_i(v)+1$ nodes change
their state. Let $d=|N_{conf}(v)|$. The subgraph $G_S$ induced by $S$
is a star graph with $d+1$ nodes and center $v$.

\begin{lemma}
  To find a lower bound for $E[Z_d]$ we may assume
  that $w$ can choose a color from $\{0,1\}\backslash tabu$ with
  $tabu=\emptyset$ if $v.final=false$ and $tabu = \{v.c\}$ otherwise
  and $v$ can choose a color from $\{0,1,\ldots,d\}\backslash tabu$
  with $tabu \subseteq \{0,1\}$.
\end{lemma}
\begin{proof}
  When a node $u\in S$ chooses a color with function {\tt randomColor}
  the color is randomly selected form $C_u= \{0,1,\ldots,
  \delta(v)\}\backslash tabu$. Thus, if $w$ and $v$ choose colors in
  the same round, the probability that the chosen colors coincide
  is \[\frac{|C_w\cap C_v|}{|C_w||C_v|}.\] This value is maximal if
  $|C_w\cap C_v|$ is maximal and $|C_w||C_v|$ is minimal. This is
  achieved when $C_w\subseteq C_v$ and $C_v$ is minimal (independent
  of the size of $C_w$) or vice versa. Thus, without loss of
  generality we can assume that $C_w\subseteq C_v$ and both sets are
  minimal. Thus, for $w\in N_{conf}(v)$ the nodes in $N(w)\backslash
  \{v\}$ already use all colors from $\{0,1,\ldots,\delta(v)\}$ but
  $0$ and $1$ and all nodes in $N(v)\backslash N_{conf}(v)$ already
  use all colors from $\{0,1,\ldots,\delta(v)\}$ but $0,1,\ldots, k$.
  Hence, a node $w\in N_{conf}(v)$ can choose a color from
  $\{0,1\}\backslash tabu$ with $tabu=\emptyset$ if $v.final=false$
  and $tabu = \{v.c\}$ otherwise. Furthermore, $v$ can choose a color
  from $\{0,1,\ldots,k\}\backslash tabu$ with $tabu \subseteq
  \{0,1\}$. In this case $tabu=\emptyset$ if $w.final=false$
  for all $w\in N_{conf}(v)$.
\end{proof}


Thus, in order to bound the expected number of rounds to reach a
legitimate state after a memory corruption we can assume that $G=G_S$
and $u.final=true$ and $w.c=0$ (i.e. $c_f=0$) for all $u\in S$. After
one round $u.final=false$ for all $u\in S$. To compute the expected
number of rounds to reach a legitimate state an execution of the
algorithm for the graph $G_s$ is modeled by a Markov chain ${\cal M}$
with the following states ($I$ is the initial state).
\begin{description}[style=multiline,leftmargin=7mm,font=\normalfont]
\item[$I$:] Represents the faulty state with $u.c=0$ and
  $u.final=true$ for all $u\in S$.
\item[$C^i$:] Node $v$ and exactly $d-i$ non-center nodes will
  not be in a legitimate state after the following round ($0\le i \le
  d$). In particular $v.final=false$ and $w.c =v.c\not=\bot$ or
  $v.c=w.c = \bot$ for exactly $d-i$ non-center nodes $w$.
    \item[$P$:] Node $v$ has not reached a legitimate state but
      will do so in the next round. In particular $v.final=false$
      and $v.c\not= w.c$ for all non-center nodes $w$.
    \item[$F$:] Node $v$ is in a legitimate state, i.e.
      $v.final=true$ and $v.c\not= w.c$ for all non-center nodes $w$,
      but $w.c$ may be equal to $\bot$.
\end{description}

${\cal M}$ is an absorbing chain with $F$
being the single absorbing state. Note that when the system is in
state $F$, then it is not necessarily in a legitimate state. This
state reflects the set of configurations considered in the last
section. 
\begin{lemma}\label{lem:trans_prob}
The transition probabilities of ${\cal M}$ are as follows:

\begin{description}[style=multiline,leftmargin=20mm,font=\normalfont]
\item[$I\longrightarrow  P$:] $\frac{d-1}{2d} + \frac{1}{d}\left(\frac{1}{2}\right)^{d+1}$
\item[$I\longrightarrow C^0$: ]
    $\frac{d-1}{d}\left(\frac{1}{2}\right)^{d+1} + \frac{1}{2d}$
    \item[$I\longrightarrow
      C^{j}$:]$\binom{d}{d-j}\left(\frac{1}{2}\right)^{d+1}$ ($0 < j
      \le d$)
    \item[$C^{i}\longrightarrow
      C^{j}$:]$\binom{d-i}{d-j}\left(\frac{1}{2}\right)^{d-i+1} +
      \frac{1}{d-i+1} \binom{d-i}{j-i}\left(\frac{1}{4}\right)^{d-i}
      (3^{d-j}-2^{d-j})$ ($0 \le i \le j \le d$)
    \item[$C^{i}\longrightarrow
      P$:]$\frac{1}{d-i+1}\left(\frac{3}{4}\right)^{d-i} +
      \frac{d-i-1}{2(d-i+1)}$ ($0 \le i < d$)
\item[$C^{d}\longrightarrow P$:] $1/2$
    \item[$P\longrightarrow F$:]$1$
\end{description}
\end{lemma}

\begin{proof}  We consider each case separately. \vspace*{-2mm}
\subsubsection*{ $I\longrightarrow P:$} Note that $u.final=true$ and $u.c=0$
  for all $u\in S$.  \\ Case 0: $v.c=\bot$. Impossible.\\
  Case 1: $v.c=0$. Impossible, since  non-center nodes
  have $c=0$ and $final=true$.\\
  Case 2: $v.c=1$. This happens with probability $\frac{1}{2d}$. All
  non-center nodes $w$ choose $w.c=\bot$, this happens with
  probability $\left(\frac{1}{2}\right)^d$.\\
  Case 3: $v.c>1$. This happens with probability $\frac{d-1}{2d}$.
  Non-center nodes can make any choice. This gives the total
  probability for this transition as
$\frac{d-1}{2d} + \frac{1}{d}\left(\frac{1}{2}\right)^{d+1}.$
\vspace*{-3mm}
\subsubsection*{$I\longrightarrow C^0:$} Note that $u.final=true$ and $u.c=0$
  for all $u\in S$.\\ Case 0: $v.c=\bot$. Non-center nodes
  choose $c=\bot$. Case has probability $\left(\frac{1}{2}\right)^{d+1}$\\
  Case 1: $v.c=0$. Impossible (see transition $I\longrightarrow P$).\\
  Case 2: $v.c=1$. At least one non-center nodes $w$ choose $w.c=1$,
  all others choose $w.c=\bot$. This
  case has probability $\frac{1}{2d}\sum_{l=1}^d
  \binom{d}{l}\left(\frac{1}{2}\right)^{d} =
  \frac{1}{2d}\left(\frac{1}{2}\right)^{d} (2^d-1)$\\
  Case 3: $v.c>1$. This case is impossible.
\vspace*{-3mm}
\subsubsection*{$I\longrightarrow C^{j}:$} Note that $u.final=true$ for all $u\in
    S$.\\ Case 1: $v.c=0$. This happens with probability
    $1/2(d-i+1)$. None of the $d-i$ non-center nodes $w$ sets
    $w.c=0$, this has probability $\left(\frac{3}{4}\right)^{d-i}$\\
    Case 2: $v.c=1$. Similar to case 1.\\
    Case 3: $v.c > 1$. (requires $d-i> 1$). This happens with
    probability $\frac{d-i-1}{2(d-i+1)}$. Non-center nodes can make any
    choice.
\vspace*{-3mm}
\subsubsection*{$C^i\longrightarrow P:$} Note that $u.final=false$ for
all $u\in S$.\\ Case 1: $v.c=0$. This happens with probability
$1/2(d-i+1)$. None of the $d-i$ non-center nodes $w$ sets
$w.c=0$, this has probability $\left(\frac{3}{4}\right)^{d-i}$\\
Case 2: $v.c=1$. Similar to case 1.\\
Case 3: $v.c > 1$. This happens with probability
$\frac{d-i-1}{2(d-i+1)}$. Note $d>i$. Non-center nodes can make any
choice.
\vspace*{-3mm}
\subsubsection*{$C^{i}\longrightarrow C^{j}:$} Note that $u.final=false$ for
    all $u\in S$.\\ Case 1: $v.c=\bot$. This happens with probability
    $\frac{1}{2}$. $d-j$ non-center nodes choose $c=\bot$ (with
    probability $\left(\frac{1}{2}\right)^{d-j}$), the other $j-i$
    non-center nodes choose $c\not=\bot$ (with probability
    $\left(\frac{1}{2}\right)^{j-i}$). The total probability for this
    case is $\binom{d-i}{d-j}\left(\frac{1}{2}\right)^{d-i+1}$.\\
    Case 2: $v.c=0$. This happens with probability
    $\frac{1}{2(d-i+1)}$. Exactly $j-i$ non-center nodes choose
    $c=1$ (with probability $\left(\frac{1}{4}\right)^{j-i}$), $1\le
    l\le d-j$ non-center nodes  choose $c=0$ (with probability
    $\left(\frac{1}{4}\right)^{l}$) and all other non-center nodes
    choose $c=\bot$ (with probability
    $\left(\frac{1}{2}\right)^{d-j-l}$). The total probability for
    this case is
\begin{eqnarray*}
     \frac{1}{2(d-i+1)}\binom{d-i}{j-i}\left(\frac{1}{4}\right)^{j-i}\sum_{l=1}^{d-j}
     \binom{d-j}{l}\left(\frac{1}{4}\right)^{l}\left(\frac{1}{2}\right)^{d-j-l}&
     =
     \\\frac{1}{2(d-i+1)}\binom{d-i}{j-i}\left(\frac{1}{4}\right)^{j-i}\sum_{l=1}^{d-j}
     \binom{d-j}{l}\left(\frac{1}{2}\right)^{d-j+l}& =
     \\\frac{1}{2(d-i+1)}\binom{d-i}{j-i}\left(\frac{1}{4}\right)^{j-i}\left(\frac{1}{2}\right)^{d-j}\sum_{l=1}^{d-j}
     \binom{d-j}{l}\left(\frac{1}{2}\right)^{l}& =
\end{eqnarray*}
\begin{eqnarray*}
     \frac{1}{2(d-i+1)}\binom{d-i}{j-i}\left(\frac{1}{2}\right)^{d+j-2i}\sum_{l=1}^{d-j}
     \binom{d-j}{l}\left(\frac{1}{2}\right)^{l}& =
     \\\frac{1}{2(d-i+1)}\binom{d-i}{j-i}\left(\frac{1}{2}\right)^{d+j-2i}
      \left(\left(\frac{3}{2}\right)^{d-j} -1\right)& =   \\\frac{1}{2(d-i+1)}\binom{d-i}{j-i}\left(\frac{1}{4}\right)^{d-i}
      \left(3^{d-j} - 2^{d-j}\right)
  \end{eqnarray*}
    Case 2: $v.c=1$. Similar to Case 1.\\
    Case 3: $v.c>0$. This does not lead to $C^{j}$ but to $P$.
\end{proof}

We first calculate the expected number $E[A_d]$ of rounds to reach the
absorbing state $F$. With Lemma~\ref{lem:cor_mess} this will enable us
to compute the expected number $E[Z_d]$ of rounds required to reach a
legitimate system state. To build the transition matrix $P$ of ${\cal
  M}$ the $d+4$ states are ordered as \[I,C^0,C^1,\ldots,C^d,P,F\] Let
$Q$ be the $(d+3) \times (d+3)$ upper left submatrix of $P$. For
$s=-1,0,1,\ldots, d+1$ denote by $Q_s$ the $(s+2) \times (s+2)$ lower
right submatrix of $Q$, i.e. $Q = Q_{d+1}$. Denote by $N_s$ the
fundamental matrix of $Q_s$ (notation as introduced in
section~\ref{sec:absorb-mark-chains}). Let $1_s$ be the column vector
of length $(s+2)$ whose entries are all $1$ and $\epsilon_s = N_s
1_s$. For $s=0,\ldots,d$ is $\epsilon_s$ the expected number of rounds
to reach state $F$ from state $C^{s}$ and $\epsilon_{d+1}$ is the
expected number of rounds to reach state $F$ from $I$, i.e.
$\epsilon_{d+1}=E[A_d]$ (Theorem 3.3.5, \cite{Kemmeny:1976}).

\begin{lemma}\label{lem:bound5}
  The expected number $E[A_d]$ of rounds to reach $F$ from $I$ is less than
  $5$ and the variance is less than $3.6$.
\end{lemma}
\begin{proof}
Note that $Q_s$ and $N_s$ are upper triangle matrices.  Let 
\[E_i - Q_i = \left(
\begin{array}{cccc}
  1-a_{1} & -a_{2} & \ldots & -a_{i+2} \\  \\[-9pt]
  0 & \multicolumn{3}{c}{\multirow{3}{*}{\Large $E_{i-1}-Q_{i-1}$}} \\[-4pt]
  \vdots & \\
  0 &
\end{array}
\right) \hspace*{9mm}
N_i = \left(
\begin{array}{cccc}
  x_{1} & x_{2} & \ldots & x_{i+2} \\  \\[-9pt]
  0 & \multicolumn{3}{c}{\multirow{3}{*}{\Large $N_{i-1}$}} \\[-4pt]
  \vdots & \\
  0 &
\end{array}\right)\]

\noindent
$E_i = (E_i-Q_i)N_i$ gives rise to $(i+2)^2$ equations. Summing up
the $i+2$ equations for the first row of $E_i$ results in
\begin{eqnarray}\label{eq:main}
\epsilon_{i} = (1-a_1)^{-1}\left(1 + \sum_{l=2}^{i+2}
  a_l\epsilon_{i+1-l}\right)
\end{eqnarray}
It is straightforward to
verify that $\epsilon_{-1} = 1$ and $\epsilon_{0} = 3$. Hence
\[ \epsilon_{i} = (1-a_1)^{-1}\left(1 + \sum_{l=2}^{i}
  a_l\epsilon_{i+1-l} + 3a_{i+1} + a_{i+2}\right)\]
Next we show by induction on $i$ that $\epsilon_i\le 4$ for
$i=-1,0,1,\ldots,d$. So assume that $\epsilon_l\le 4$ for
$l=-1,0,1,\ldots, i-1$ with $i < d$. Then
\[ \epsilon_{i} \le (1-a_1)^{-1}\left(1 + 4\sum_{l=2}^{i}
  a_l + 3a_{i+1} + a_{i+2}\right)\]
since $a_i\ge 0$. Using the fact $1-a_1 = \sum_{l=2}^{i+2}a_l$ this inequality becomes
\begin{eqnarray*}
\epsilon_{i} &\le& (1-a_1)^{-1}(1 + 4 (1-a_1-a_{i+1}-a_{i+2}) +
3a_{i+1} + a_{i+2})
 =  4 + \frac{1 -a_{i+1}-3a_{i+2}}{1-a_1}
\end{eqnarray*}
Coefficient $a_j$ denotes the transition probability from $C^{d-i}$ to
$C^{d+j-(i+1)}$ for $j=1,\ldots,i+1$ and $a_{i+2}$ that for changing
from $C^{d-i}$ to $P$. For $i\le d$ the following values from
Lemma~\ref{lem:trans_prob} are used:
\begin{eqnarray*}
a_{1} &=& \left(\binom{i}{i+1-l}\left(\frac{1}{2}\right)^{i+1} +
      \frac{1}{i+1} \binom{i}{l-1}\left(\frac{1}{4}\right)^{i}
      (3^{i+1-l}-2^{i+1-l})\right)\\
a_{i+1}& = & \left(\frac{1}{2}\right)^{i+1} \text{ and }
a_{i+2} =  \frac{1}{i+1}\left(\frac{3}{4}\right)^{i} +
      \frac{i-1}{2(i+1)}. \text{ ~ Thus,}
\end{eqnarray*}
 \[3a_{i+2} =
\frac{3}{i+1}\left(\frac{3}{4}\right)^{i} + \frac{3(i-1)}{2(i+1)} >
1\] holds for $i\ge 2$. This yields \[\frac{1
  -a_{i+1}-3a_{i+2}}{1-a_1} < 0\] and therefore $\epsilon_i \le 4$.
To bound $\epsilon_{d+1}$ we use Equation~\ref{eq:main}
with $i=d+1$. Note that in this case $a_1=0$ since a transition from
$I$ to itself is impossible. Hence
\begin{eqnarray*}
  E[A_d]= \epsilon_{d+1} = 1 + \sum_{l=2}^{d+3}  a_l\epsilon_{d+2-l} \le 1 + 4
  \sum_{l=2}^{d+3}  a_l = 5
\end{eqnarray*}
Thus, $Var[A_d] = ((2N_{d+1} - E_{d+1})1_{d+1} - 1_{d+1}^2)[1] =
2\sum_{i=1}^{d+3} x_i\epsilon_{d+2-i} - \epsilon_{d+1} -
\epsilon_{d+1}^2$
Fig.~\ref{fig:E_Var} shows that $Var[A_d]\le 3.6$.
\end{proof}


\begin{lemma}\label{lem:exp_memory}
  The expected value for the containment time after a
  memory corruption at node $v$ is at most  $\frac{1}{\ln 2}
  H_{\delta_i(v)}+11/2$ with variance  less than $7.5$.
\end{lemma}
\begin{proof}
  For a set $X$ of configurations and a single system
  configurations $c$ denote by $E(c,X)$ the expected value for the
  number of transitions from $x$ to a state in $X$. Denote by ${\cal
    L}$ the set of legitimate system states. Then

\begin{eqnarray*}
  E(I, {\cal L}) &=& \sum_{e\in T(I, {\cal L})}
  l(e)p(e)\\
&=& \sum_{x\in F} \sum_{e_1\in T(I, x)} \sum_{e_2\in T(x
 , {\cal L})} (l(e_1) + l(e_2))p(e_1)p(e_2)\\
&=& \sum_{x\in F} \sum_{e_1\in T(I, x)} \left(l(e_1)p(e_1)\sum_{e_2\in T(x
 , {\cal L})} p(e_2) + p(e_1)\sum_{e_2\in T(x
 , {\cal L})}l(e_2)p(e_2)\right)\\
&=& \sum_{x\in F} \sum_{e_1\in T(I, x)} \left(l(e_1)p(e_1) + p(e_1)E(x, {\cal L})\right)
\end{eqnarray*}
\begin{eqnarray*}
&=& \sum_{x\in F} \left(E(I,x) + \sum_{e_1\in T(I
   , x)} p(e_1)E(x, {\cal L})\right)\\
&\le& E(I, F) + \max \{E(x, {\cal L}) \mid x\in F\} \sum_{e_1\in T(I
   , F)} p(e_1)\\
&=&E(I, F) + \max \{E(x, {\cal L}) \mid x\in F\} \le 5 + \frac{1}{\ln 2}
  H_{\delta_i(v)}+1/2
\end{eqnarray*}

The last step uses Lemma~\ref{lem:bound5} and
\ref{lem:cor_mess}. The bound on the variance is proved similarly.
\end{proof}

Theorem~\ref{theo:stab_time}, Lemma~\ref{lem:exp_memory},
\ref{lem:cor_mess}, and \ref{lem:contain} prove the following
Theorem.

\begin{theorem} \label{theo:contain} \AC is a self-stabilizing
  algorithm for computing a $(\Delta+1)$-coloring in the synchronous
  model within $O(\log n)$ time with high probability. It uses
  messages of size $O(\log n)$ and requires $O(\log n)$ storage per
  node. With respect to memory and message corruption it has
  contamination radius $1$. The expected containment time is at most
  $\frac{1}{\ln 2} H_{\Delta_i}+11/2$ with variance less than $7.5$.
\end{theorem}

\begin{corollary}
  Algorithm \AC has expected containment time $O(1)$ for
  bounded-indepen\-dence graphs. For unit disc graphs this time is at
  most $8.8$.
\end{corollary}
 \begin{proof}
For these graphs $\Delta_i\in O(1)$, in particular
$\Delta_i\le 5$ for unit disc graphs.
 \end{proof}

\section{Conclusion}
\label{sec:conclusion}

This paper presented techniques to derive upper bounds for the mean
time to recover of a single fault for self-stabilizing algorithms in
the synchronous message passing model. For a new $\Delta+1$-coloring
algorithm we analytically derive a bound of $\frac{1}{\ln 2}
H_{\Delta_i}+11/2$ for the expected containment and showed that the
variance less than $7.5$. We believe that the technique can also be
applied to other self-stabilizing algorithms.

\subsection*{Acknowledgments} Research was funded by Deutsche
Forschungsgemeinschaft DFG (TU 221/6-1).

\bibliographystyle{plain}
\bibliography{markov}
\end{document}